\documentclass[12pt]{article}
\usepackage{amssymb}
\usepackage{amsmath}
\usepackage{epsf}
\newcommand{\be}{\begin{equation}}
\newcommand{\ee}{\end{equation}}
\newcommand{\R}{{\mathbb R}}
\newcommand{\E}{{\mathbb E}}
\newcommand{\PP}{{\mathbb P}}
\newcommand{\N}{{\mathbb N}}
\newcommand{\xx}{{\bf x}}
\newcommand{\yy}{{\bf y}}

\newcommand{\ur}{{\bf r}}
\newcommand{\us}{{\bf s}}
\newcommand{\ut}{{\bf t}}
\newcommand{\vr}{{\vec{\bf r}}}
\newcommand{\vs}{{\vec{\bf s}}}
\newcommand{\vt}{{\vec{\bf t}}}

\newtheorem{thm}{Proposition}

\newenvironment{proof}{\noindent{\bf Proof:} }{\hfill $\Box$ \\}
\begin{document}
\pagestyle{myheadings}
\thispagestyle{empty}
\setcounter{page}{1}

\begin{center}
{\LARGE\sc 
State estimation for temporal point processes}\\[0.5in]
{M. N. M. van Lieshout}\\[0.1in]
{\footnotesize\em CWI
\\ 
P.O. Box 94079, NL-1090 GB Amsterdam, The Netherlands \\[0.1in]
Department of Applied Mathematics,
University of Twente \\
P.O. Box 217, NL-7500 AE Enschede, The Netherlands }\\[0.5in]
\end{center}

\begin{verse}
{\footnotesize
\noindent
This paper is concerned with combined inference for point processes on the
real line observed in a broken interval. For such processes, the classic
history-based approach cannot be used. Instead, we adapt tools from sequential 
spatial point processes. For a range of models, the marginal and conditional 
distributions are derived. We discuss likelihood based inference as well 
as parameter estimation using the method of moments, conduct a simulation 
study for the important special case of renewal processes and analyse a 
data set collected by Diggle and Hawtin.\\[0.2in]

%Point process
%Renewal processes
%Inference stochastic processes
\noindent
{\em Key words \& Phrases:} Cox process, Markov chain Monte Carlo,
Markov point process, renewal process, sequential point process,
state estimation.
}
\end{verse}

\section{Introduction}

Inference for point processes on the real line has been dominated
by a dynamic approach based on the stochastic intensity 
\cite{Brem72,Karr91,Last95} which expresses the likelihood of
a point at any given time conditional on the history of the 
process. Such an approach is quite natural
in that it is the mathematical translation of the intuitive idea 
that more and more information becomes available as time passes.
Furthermore, the approach allows the utilisation of powerful 
tools from martingale theory. In a statistical sense, since the
stochastic intensity is closely related to the hazard rates of the
distribution of the length of the inter-point intervals, a 
likelihood is immediately available \cite{DVJ-I}.

The dynamic approach, however, does not seem capable of dealing 
with situations in which the flow of time is interrupted. In such 
cases, combined state estimation techniques are needed that are 
able to simultaneously carry out inference and to reconstruct the 
missing points. More specifically, state estimation aims to find 
`the optimal reconstruction, realization by realization, of 
unobserved portions of a point process or of associated random 
variables or processes' \cite[p.~93]{Karr91}.  Other examples include
extrapolation and prediction \cite{Moll04}, filtering \cite{Sing09}, 
cluster detection \cite{Lies02} or the estimation of the driving 
measure of a Cox process \cite{Moll98}.

The aim of this paper is to apply ideas from sequential point 
process \cite{Lies06a,Lies06b}, in particular the sequential
Papangelou conditional intensity which describes the probability
of finding a point at a particular time conditional on the 
remainder of the process. Thus, the concept is related to the
stochastic intensity, except that the future is taken into 
account as well as the past. It is this last feature that allows
the incorporation of missing data. Moreover, for hereditary point 
processes at least, they define a likelihood.

The plan of this paper is as follows. In Section~\ref{S:prelim}, we 
represent a point processes on the real line as a sequential point
process on semi-cubes and recall the definition of the sequential
Papangelou conditional intensity. Section~\ref{S:broken} derives
the marginal and conditional distributions for the situation where
the point process is observed in disjoint intervals only. Important
special cases are studied in detail: Markov point processes in 
Section~\ref{S:fixed}, renewal processes in Section~\ref{S:near}
and Cox processes in Section~\ref{S:Cox}. In Section~\ref{S:simulation},
we show that a substantial bias may be incurred when missing data
is not accounted for properly. We present a  Metropolis--Hastings 
algorithm for sampling from the conditional distribution of the 
missing data as well as a locally defined birth-and-death process
based on the sequential Papangelou conditional intensity in
Section~\ref{S:Geyer} and the Appendix. These techniques are then 
used in Monte Carlo inference. Finally, in Section~\ref{S:Diggle},
a health care surveillance data set collected by Diggle and Hawtin 
is analysed to illustrate the approach and we conclude with a 
summary.

\section{Finite point process on the line}
\label{S:prelim}

In this section, we review basic concepts from the theory of 
chronologically ordered, simple point processes on an interval
$[0, T]$, $T > 0$. 

Realisations of a simple point process $X$ consist of a finite
number of distinct points
\(
\{ t_1, \dots, t_n \} \subset [0, T]
\) 
for $n\in \N_0$. Since the points are naturally ordered, there
is a unique correspondence between the set $\{ t_1, \dots, t_n \}$ 
and the vector $(t_1, \dots, t_n)$ where $t_1 < t_2 < \cdots < t_n$ 
\cite{DVJ-I,Last95}. The special case $n=0$ corresponds to an
empty realisation. Equipping $H_n([0,T]) = \{ (t_1, \dots, t_n) \in 
[0,T]^n: t_1 < \cdots < t_n \}$, $n\in\N$, with the product Borel 
$\sigma$-algebra, the distribution of a finite sequential point 
process $Y$ can be specified by \cite{Lies06a}.
\begin{itemize}
\item a probability mass function $q_n$, $n\in\N_0$, for the  
number of points in $[0, T]$;
\item for each $n\in \N$ for which $q_n>0$, a probability density 
function $p_n$ on $H_n([0,T])$ for the chronologically ordered vector 
of point locations given that there are $n$ of them.
\end{itemize}
Note that the $p_n$ need not be symmetric.

The likelihood of finding $n\in \N$ points, one at each of the locations 
$t_1, \dots, t_n$ in that order, is expressed by the Janossy densities
\cite{DVJ-I}
\[
j_n(t_1, \dots, t_n) = q_n p_n( t_1, \dots, t_n)
\]
on $H_n(T)$. For $n=0$, $j_0 = q_0$. 
More compactly, one may specify a density $f$ on 
$N^f([0,T]) = \cup_{n=0}^\infty H_n([0,T])$ (equipped with the $\sigma$-algebra
generated by the Borel product $\sigma$-fields) with respect to the 
probability measure
\[
\nu(F) = \nu_{[0,T]}(F) = \sum_{n=0}^\infty e^{-T} \int_0^T \int_{t_1}^T
\cdots \int_{t_{n-1}}^T 1_F( t_1, \dots, t_n) \, dt_1 \cdots dt_n
\]
which is related to the unnormalised Janossy densities as
\(
f(t_1, \dots, t_n) = e^T j_n(t_1, \dots, t_n)
\)
and $f(\emptyset) = e^T q_0$.
Note that $\nu_{[0,T]}$ is the distribution of a unit rate Poisson
process on $[0,T]$.

The sequential Papangelou conditional intensity \cite{Lies06b} for 
inserting $t$ at position $k \in \{ 1, \dots, n+1 \}$ is defined by
\begin{equation}
\label{e:Papa}
\lambda_k(t | t_1, \dots, t_n) =
\frac{ f( t_1, \dots, t_{k-1}, t, t_k, \dots, t_n)}{
 f(t_1, \dots, t_n)}
\end{equation}
whenever both $t_{k-1} < t < t_k$ and $f(t_1, \dots, t_n) > 0$; it is
set to zero otherwise. Reversely, provided $f(\cdot)$ is hereditary
in the sense that $f( t_1, \dots, t_n) > 0$ implies that $f( s_1, 
\dots, s_m) > 0 $ for all subsequenses $(s_1, \dots, s_m)$ of 
$(t_1, \dots, t_n)$, the density factorises as
\begin{equation}
\label{e:f-Papafactor}
f( t_1, \dots, t_n) = f(\emptyset) \prod_{i=1}^n \lambda_i(t_i |
t_1, \dots, t_{i-1}).
\end{equation}

At this point, it is important to note that the term Papangelou
conditional intensity is often abbreviated to plain `conditional
intensity'. We choose not to do so to avoid confusion with 
another factorisation often referred to by this name. Indeed, by
\cite[Prop.~7.2.III]{DVJ-I}, the probability density function
of $Y$ can be factorised as
\[
f( t_1, \dots, t_n ) =  e^T \prod_{i=1}^n h^*(t_i )
\exp\left[ - \int_0^T h^*(t ) dt \right]
\]
where $h^*$ is a combination of the hazard functions of the conditional
probability density functions $\pi_n( t_n | t_1, \dots, t_{n-1})$ that
govern the location of the $n$-th point conditional on the `past' points
$t_1, \dots, t_{n-1}$. We shall refer to $h^*$ as the `hazard
rate'. Note that one sometimes writes $h^*(t_i) =  h^*(t_i | t_1, \dots, 
t_{i-1})$ to stress the dependence on the history of $Y$ up to but not 
including $t_i$. The two concepts are related, but not identical in general. 
For a rigorous treatment, the reader is referred to \cite{Karr91,Last95}.

%Note: in the Festschrift the ref measure \nu has extra term 1/n!
%for the random ordering. Therefore f_Festschrift is n! f_thisPaper

\section{Partially observed point processes}
\label{S:broken}

Consider a sequential point process that is observed in a broken interval 
$[0,T_1] \cup [T_2, T]$. In other words, any point falling in $(T_1, T_2)$ 
is not observed. Thus, a realisation consists of two chronologically 
ordered sequences, say $\vr = (r_1, \dots, r_k)$, $k\in\N_0$, in $[0,T_1]$ 
and $\vs = (s_1, \dots, s_l)$, $l\in \N_0$, in $[T_2, T]$.

In order to carry out likelihood based inference or state estimation,
we need, respectively, the marginal distribution of the observations
and the conditional law of the missing points. 

\begin{thm}
\label{T:interpolation}
Let $Y$ be a simple sequential point process on $[0,T]$, $T>0$, 
defined by $p_n^Y$, $q_n^Y$ as in Section~\ref{S:prelim}. 
\begin{description}
\item[(a)] The marginal distribution of the sequential point process
$Z$ defined as the restriction of $Y$ to $[0,T_1] \cup [T_2, T]$ 
has Janossy density 
\begin{equation}
\label{e:jZ}
j^Z(\vr, \vs) = \sum_{n=0}^\infty q^Y_{n+n(\vr)+n(\vs)} 
   \int_{H_n((T_1,T_2))} 
      p^Y_{n+n(\vr) + n(\vs)}( \vr, t_1, \dots, t_n, \vs ) 
         \prod_{i=1}^n dt_i 
\end{equation}
where $n(\cdot)$ denotes vector length (with the convention that the
integral for $n=0$ is equal to $p^Y_{n(\vr) + n(\vs)}(\vr, \vs)$).
In other words, $Z$ has density 
\(
f^Z(\vr, \vs) = e^{T - T_2 + T_1} j^Z(\vr, \vs)
\)
with respect to $\nu_{[0,T_1] \cup [T_2,T]}$.
\item[(b)] If $\vr$ and $\vs$ are feasible in the sense that their
marginal Janossy density (\ref{e:jZ}) is strictly positive, then the 
conditional distribution of $Y$ on $(T_1, T_2)$ given $\vr$ and $\vs$ 
has Janossy density
\begin{equation}
\label{e:jCond}
j_n( t_1, \dots, t_n | \vr, \vs ) = 
\frac{1}{j^Z(\vr, \vs)} \, q^Y_{n+n(\vr) + n(\vs)} \,
p^Y_{n+n(\vr) + n(\vs)}(\vr, t_1, \dots, t_n, \vs) 
\end{equation}
for $n\in\N$, $q_0(\vr, \vs) = q^Y_{n(\vr) + n(\vs)} / j^Z(\vr, \vs)$.
Equivalently, a conditional density with respect to $\nu_{(T_1, T_2)}$
is  
\(
f( t_1, \dots, t_n | \vr, \vs) =
e^{T_2-T_1} j_n(t_1, \dots, t_n | \vr, \vs).
\)
\end{description}
\end{thm}

\begin{proof}
Split $[0,T]$ in three parts: $[0,T_1]$, $(T_1, T_2)$ and $[T_2, T]$.
Then, 
\[
N^f = \cup_{n=0}^\infty H_n([0,T]) = \cup_{n=0}^\infty 
\cup_{k, j, l : k + j + l = n}^\infty 
 H_k([0,T_1]) \times H_j((T_1, T_2)) \times H_l([T_2, T]).
\]
Write $H_{k,j,l}(T_1, T_2, T) = 
H_k([0,T_1]) \times H_j((T_1, T_2)) \times H_l([T_2, T])$. 
Then, the $\sigma$-algebra on the disjoint union set 
$\cup_{j+k+l=n} H_{k,j,l}(T_1, T_2,T)$ can be described as follows. 
A subset $B$ is measurable exactly if $B\cap H_{k,j,l}(T_1, T_2, T)$
is a Borel set for all $k,j,l$ adding up to $n$ \cite{Frem03}.
Clearly this holds for all Borel sets in $H_n([0,T])$ and,
reversely, any $B$ measurable in the union $\sigma$-algebra
is the disjoint union of Borel sets, hence a Borel set itself.
Consequently, $Y$ can also be split into three well-defined,
though possibly dependent, sequential point processes by 
restriction to $[0,T]$, $(T_1, T_2)$ and $[T_2,T]$. 

For $k, j, l \in \N_0$, write $F_{k,j,l}$ for the event that $j$ 
points fall in $(T_1, T_2)$, $k$ points in $[0,T_1]$ and $l$ 
further points in $[T_2, T]$. Then
\[
\PP(F_{k,j,l}) = q^Y_{k+j+l}
\int_{H_k([0,T_1])} \int_{H_j((T_1, T_2))} \int_{H_l([T_2, T])}
p^Y_{k+j+l}( \vr, \vt, \vs ) d\vr \, d\vt \, d\vs
\]
when $k+j+l>0$ and $q^Y_0$ otherwise.

\begin{description}
\item[(a)] For $n\in\N_0$, the marginal probability $q_n$ of placing 
$n$ points in $[0,T_1] \cup [T_2, T]$ is the probability of the disjoint 
union $\cup \{ F_{k,j,l} : k, j, l \in \N_0, k+l=n \}$, that is,
\[
q^Z_n = \sum_{k,l \in \N_0: k+l = n} \sum_{j=0}^\infty \PP(F_{k,j,l}).
\] 
Similarly, suppose that $q^Z_n > 0$ and condition on having $n$ points 
in $[0,T_1] \cup [T_2, T]$. Then, for any measurable $A \subseteq N^f(
[0,T_1] \cup [T_2, T])$,
\begin{eqnarray*}
\PP(Z \in A | n(Z) = n ) & =  &
\sum_{k,l\in\N_0: k+l=n} \sum_{j=0}^\infty 
\PP( \{ Z \in A \} \cap F_{k,j,l} | n(Z) = n )\\
& = &
\frac{1}{q_n^Z} \sum_{k,l\in\N_0: k+l=n} \sum_{j=0}^\infty 
\PP( \{ Z \in A \} \cap F_{k,j,l} ).
\end{eqnarray*}
Since $\PP( \{ Z \in A \} \cap F_{k,j,l} )$ can be written as
\[
 q^Y_{k+j+l} \int_{H_k([0,T_1])}
\int_{H_j((T_1, T_2))} \int_{ H_l([T_2, T])} 1_A(\vr, \vs) \, 
 p^Y_{k+j+l}(\vr, \vt, \vs) \, d\vr \, d\vt \, d\vs,
\]
a particular configuration consisting of $\vr \in
H_k([0,T_1])$ and $\vs \in H_l([T_2, T])$ with $k+l=n$ 
has (conditional) probability density 
\[
\frac{1}{q_n^Z}  \sum_{j=0}^\infty q^Y_{n+j} \int_{H_j((T_1, T_2))} 
    p^Y_{n+j}( \vr, \vt, \vs ) d\vt
\]
with respect to Lebesgue measure on $H_n([0,T_1] \cup [T_2, T])$.

\item[(b)] Next, condition on $\vr \in H_k([0,T_1])$ and $\vs \in 
H_l([T_2, T])$. To compute the probability of having $n$, $n\in\N_0$, 
points in between $T_1$ and $T_2$, we may restrict
ourselves to $F_{k,n,l}$. If $\PP(F_{k,n,l}) = 0$, this probability is
zero. Thus, assume $\PP(F_{k,n,l}) > 0$ and, a fortiori, 
$q^Y_{k+n+l} > 0$. On $F_{k,n,l}$, the (joint) likelihood of the
$k+n+l$ points is $q^Y_{k+n+l} p^Y_{k+n+l}(\cdot)$, hence integration 
over $H_n((T_1, T_2))$ yields that the marginal density of $(\vr, \vs)$
with respect to Lebesgue measure on $H_k([0, T_1]) \times H_l([T_2, T])$
is given by 
\[
 q^Y_{k+n+l} \int_{H_n((T_1, T_2))} p^Y_{k+n+l}(\vr, \vt, \vs) \, d\vt.
\]
We conclude that the conditional probability of finding $n$ points in
$(T_1, T_2)$ given $\vr$ and $\vs$ elsewhere is
\[
q_n(\vr, \vs) = \frac{ q^Y_{k+n+l}  \int_{H_n((T_1, T_2))} 
  p^Y_{k+n+l}(\vr, \vt, \vs) \, d\vt
 }{
  \sum_{j=0}^\infty q^Y_{k+j+l}  \int_{H_j((T_1, T_2))} 
  p^Y_{k+j+l}(\vr, \vt, \vs) \, d\vt }.
\]
The denominator is strictly positive by assumption. Moreover, a
conditional density function for the locations of points in 
$Y \cap (T_1, T_2)$ given there are $n$ such points and given
realisations $\vr$ and $\vs$ elsewhere exists and reads
\[
\frac{p^Y_{k+n+l}(\vr, \cdots, \vs)}{
\int_{H_n((T_1, T_2))} p^Y_{k+n+l}(\vr, \vt, \vs) d\vt
}
\]
provided the denominator is strictly positive, which is always the
case if $q_n(\vr, \vs) > 0$.
\end{description}
\end{proof}

For hereditary sequential point processes, we have the following 
simplification of Proposition~\ref{T:interpolation}(b).

\begin{thm}
\label{T:interpolation-hereditary}
Let $Y$ be a simple hereditary sequential point process with density $f^Y$ 
with respect to $\nu_{[0,T]}$. Let $\vr \in H_k([0,T_1])$ and $\vs\in 
H_l([T_2,T])$ be feasible in the sense that
\(
j^Z(\vr, \vs) > 0.
\)
Then the conditional distribution of $Y$ on $(T_1, T_2)$, 
$0<T_1 < T_2 < T$, given $\vr$ and $\vs$ is hereditary with Papangelou 
conditional intensity 
\begin{equation}
\label{e:ci-hereditary}
\lambda_{n+1}( t | t_1, \dots, t_n; \vr, \vs ) =
\lambda_{n+k+1}^Y( t | \vr, t_1, \dots, t_n, \vs)
\end{equation}
equal to the Papangelou conditional intensity
of $Y$ for all $(t_1, \dots, t_n) \in H_n((T_1, T_2))$  and $t \in (t_n, T_2]$
provided $f^Y(\vr, t_1, \dots, t_n, \vs) > 0$.
\end{thm}

\begin{proof}
By the proof of Proposition~\ref{T:interpolation} part (b),
\[
q_0(\vr, \vs) = \frac{q^Y_{k+l} p^Y_{k+l}(\vr, \vs)}{ j^Z(\vr, \vs)}.
\]
Therefore $q_0(\vr, \vs) = 0$ would imply that $j^Y_{k+l}(\vr, \vs)$
and hence $f^Y(\vr, \vs)$ would be zero. 
Since $Y$ is hereditary, $f^Y( \vr, t_1, \dots, t_j, \vs )$ 
would be zero for all $t_1 < \cdots < t_j$ and all $j\in \N_0$. Consequently, 
$j^Z(\vr, \vs)$ would be zero, a contradiction with the assumption. 
We conclude that $q_0(\vr,\vs) > 0$. 

For $n\in \N$, consider the chronologically ordered vector 
$(t_1, \dots, t_n)$ with components in $(T_1, T_2)$ and assume that 
$j_n(t_1, \dots, t_n | \vr, \vs) > 0$. By~(\ref{e:jCond}) and the 
assumption that $j^Z(\vr, \vs) > 0$, 
$j^Y_{k+n+l}(\vr, t_1, \dots, t_n, \vs) > 0$. Since $Y$ is hereditary, 
all subsequences also have strictly positive Janossy density, in 
particular those including both $\vr$ and $\vs$. Another appeal to 
(\ref{e:jCond}) implies that the conditional point process is
hereditary too. 

By the general theory outlined in Section~\ref{S:prelim}, the
density factorises as 
\[
f(t_1, \dots, t_n|\vr, \vs) = e^{T_2 - T_1} q_0(\vr, \vs)
\prod_{i=1}^n \lambda_i( t_i | t_1, \dots, t_{i-1}; \vr, \vs )
\]
in terms of the Papangelou conditional intensity
\[
\lambda_{i}( t_i | t_1, \dots, t_{i-1}; \vr, \vs )  =  \frac{ 
f( t_1, \dots, t_i | \vr, \vs)
}{
f( t_1, \dots, t_{i-1} | \vr, \vs )
} = 
\frac{
f^Y( \vr, t_1, \dots, t_i, \vs ) }{ f^Y( \vr, t_1, \dots, t_{i-1}, \vs) }
\]
for $f(t_1, \dots, t_{i-1} | \vr, \vs) > 0$ and zero otherwise.
The last equality follows from Proposition~\ref{T:interpolation} part (b)
and implies (\ref{e:ci-hereditary}).
\end{proof}

The Papangelou conditional intensity, in contrast to the hazard rate 
which is a function of the history of the process, is able to take
into account past as well as future points. Indeed, there does not
seem to be an analogue of (\ref{e:ci-hereditary}) in terms of $h^*$.

The theorems above cannot be simplified without specific model assumptions.
In the next sections we therefore specialise to, respectively, 
Markov point processes, nearest--neighbour processes and Cox models.

\section{Markov point process}

In this section, we consider models in which the dependence of 
a sequential point process in $(T_1, T_2)$ on its future and past 
is limited either to a border region $[T_1 - R, T_1] \cup
[T_2, T_2 + R]$ for some fixed interaction range $R$ 
(Section~\ref{S:fixed})
or by its two nearest neighbours, say $r_k$ in $[0, T_1]$ and 
$s_1 \in [T_2, T]$ (Section~\ref{S:near}).

\subsection{Fixed range dependence}
\label{S:fixed}

A chronologically ordered sequential point process $Y$ on $[0,T]$
is said to be Markov at range $R > 0$ if it is hereditary and
and, for all $( t_1,  \dots, t_n )$ for which $f^Y(t_1, \dots, t_n)
> 0$ and for all $u \neq t_i$, $i= 1, \dots, n$, the Papangelou conditional 
intensity $\lambda^Y_{n+1}(u | t_1, \dots, t_n)$ depends only on $u$ and 
the set of points $t_i < u $ that have temporal distance $ u - t_i 
\leq R$. In other words, the underlying unordered point process is Markov
in the classical sense \cite{Ripl77}. Therefore, by the 
Hammersley--Clifford theorem, the density of $Y$ factorises as
\[
f^Y(t_1, \dots, t_n)  = 1 \{ t_1 < \cdots < t_n \} \prod \psi( \yy )
\]
for some non-negative integrable interaction function $\psi$, 
ranging over the collection of {\em unordered} finite subsets $\yy$ of 
$\{ t_1, \dots, t_n\} \subseteq [0,T]$, that vanishes except on cliques. 
More precisely, $\psi(\yy) = 1$ unless it consists of pairwise $R$-close 
points. To get rid of the indicator function, define
\(
 \varphi( u, \yy ) = \psi( \yy \cup \{ u \} ) 1\{ u > \max( \yy ) \}
\)
for $u\in [0, T]$ so that
\[
f^Y(t_1, \dots, t_n) =
 f^Y(\emptyset)  \prod_{i=1}^n 
 \prod_{\ut \subseteq \{ t_1, \dots, t_{i-1} \}}
\varphi^Y( t_i, \ut) 
\]
cf.\ \cite{Lies06a}. As usual, for $n = 0$, $(t_1, \dots, t_0)$ 
is interpreted as the empty sequence.

Next, turn to (\ref{e:ci-hereditary}) and observe that it reduces to
\[
\lambda_{n+1}( t | t_1, \dots, t_n ; \vr, \vs) =
\prod_{\yy \subseteq \{ t_1, \dots, t_n \} } \left[ 
\psi( \{ t \} \cup \yy )
\prod_{\emptyset \neq \xx \subseteq \ur \cup \us }
\psi( \{  t \} \cup \xx \cup \yy )
\right],
\]
the product of all interaction functions that involve the point $t$.
We conclude that the Markov property is preserved under conditioning 
but the clique interaction functions 
\[
  \varphi( t, \yy | \vr, \vs ) =  \varphi^Y( t, \yy )
     \prod_{\emptyset \neq \xx \subseteq \ur \cup \us }
          \psi( \{  t \} \cup \xx \cup \yy )
\]
for $\yy \subseteq \{ t_1, \dots, t_n \}$ and $t > t_n$ 
may change and depend on points of $\vr$ and $\vs$ in 
$[T_1 - R, T_2 + R]$. In particular, since 
$\varphi(t, \emptyset | \vr, \vs)$ need not be constant in
$t$ even when $\varphi^Y(t, \yy)$ is, non-homogeneity may be introduced
by the conditioning.

The special case of pairwise interaction processes of the form
\begin{equation}
\label{e:pair}
f^Y( t_1, \dots, t_n ) = f^Y(\emptyset)\prod_{i=1}^n \varphi^Y(t_i, \emptyset) 
\prod_{  i < j: t_j - t_i \leq R} \varphi^Y(t_j, \{ t_i \})
\end{equation}
with $(t_1, \dots, t_n ) \in H_n([0,T])$ is of special interest.
For such models, the conditional distribution of $Y$ on $(T_1, T_2)$, 
$0<T_1 < T_2 < T$, given $\vr$ on $[0, T_1]$ and $\vs$ on $[T_2, T]$ has 
first and second order interaction functions
\begin{eqnarray*}
\varphi(t, \emptyset | \vr, \vs ) & = &
\varphi^Y( t , \emptyset) \prod_{x \in \vr: t-x \leq R} \varphi^Y( t, \{ x \} ) 
\prod_{x \in \vs: x-t \leq R} \varphi^Y ( x, \{ t \} ) ;
\\
\varphi( t_2, \{ t_1 \} | \vr, \vs ) & = & \varphi^Y( t_2, \{ t_1 \} ) .
\end{eqnarray*}
Therefore the conditional point process is also Markov at range $R$ with
pairwise interactions only. The second order interaction function is
unaffected by the conditioning. The first-order interaction function
depends on points of $\vr$ and $\vs$ up to a range $R$ from $(T_1, T_2)$.

\subsection{Nearest neighbour dependence}
\label{S:near}

A renewal process on $[0, T]$ is defined as follows \cite[Chapter~8]{Karr91}.
Starting at time $0$, let $U_i$ be a sequence of independent and identically 
distributed inter-arrival times with common probability distribution 
function $F_\pi$. We assume that $F_\pi$ is non-defective and absolutely 
continuous with density $\pi$. Its hazard function is denoted by $h_\pi$.
Set $T_0 = 0$ and $T_i = T_{i-1} + U_i$ for $i\in\N$. Then those $T_i$,
$i \geq 1$ falling in $(0,T]$ form a simple sequential point process $Y$.
For simplicity, we assume that the process starts at time zero, but
other starting distributions may also be accommodated 

Due to the independence assumptions in the model, the hazard rate of
$Y$ is particularly appealing. Write $V_t$ for the backwards recurrence 
time at $t$, that is, the difference between $t$ and the last event 
falling before or at time $t$. Then, $h^*(t) = h_\pi(V_{t-})$ \cite{Karr91}, 
so that the density $f^Y$ with respect to $\nu_{[0,T]}$ is given by
\begin{eqnarray}
\nonumber
f^Y(t_1, \dots, t_n) & = &
e^T \prod_{i=1}^n h^*(t_i )
\exp\left[ - \int_0^T h^*(t ) dt \right]
\\
& = & e^T  (1 - F_\pi(T-t_n)) 
\prod_{i=1}^{n} \pi(t_i - t_{i-1})
\label{e:fY-renewal}
\end{eqnarray}
for $(t_1, \dots, t_n) \in H_n([0,T])$ \cite{DVJ-I,Karr91}, under
the conventions that an empty product is set to one and that $t_0 = 0$.

The sequential point process $Y$ is not necessarily hereditary, 
for example when $\pi$ has small bounded support. Thus, we shall
assume that $\pi > 0$. In this case, $f^Y$ can be factorised in terms of 
its Papangelou conditional intensity
\begin{equation}
\label{e:ci-renewal}
\lambda^Y_{n+1}(t | t_1, \dots, t_n) = 
 \pi(t - t_{n} )  \frac{ 1 - F_\pi(T-t) }
{1 - F_\pi(T-t_n) }
\end{equation}
for $t_n < t \leq T$. It depends on the configuration 
$t_1, \dots, t_n$ only through the point $t_n$ that is closest to $t$, 
no matter how large the distance $t-t_n$. Such models are known as 
nearest-neighbour Markov point processes \cite{Badd89}. Their density
satisfies a factorisation similar to (\ref{e:pair}) for fixed range 
Markov processes. Indeed, for $n \in \N$,
\[
f^Y( t_1, \dots, t_n ) = f^Y(\emptyset) 
 \prod_{i=1}^n \varphi^Y(t_i, \emptyset )
\prod_{i=2}^n \varphi^Y( t_i, \{ t_{i-1} \} | t_1, \dots, t_n )
\]
and $f^Y(\emptyset) = e^T ( 1 - F_\pi(T) )$. The interaction functions
are
\[
\varphi^Y( t , \emptyset) =  \frac{ \pi(t ) ( 1 - F_\pi(T-t) ) }{
  1 - F_\pi(T) }
\]
and
\[
\varphi^Y( u, \{ t \} | t_1, \dots, t_n )  =  
  \frac{\pi(u - t) \left[ 1 - F_\pi(T) \right]
   }{
        \pi( u ) \left[  1 - F_\pi(T-t) \right] }
\]
for $u\in \{ t_1, \dots, t_n \}$ and $t = \max\{ t_i : t_i < u \}$ with
$\varphi^Y(t, \yy | t_1, \dots, t_n ) = 1$ otherwise.
Note that, as the pairwise interaction function involves two
consecutive neighbours, it depends on the 
whole configuration in contrast to the models discussed
in the previous subsection.

Next, consider the conditional distribution of $Y$ on $(T_1, T_2)$ 
given $\vr$ on $[0,T_1]$ and $\emptyset \neq \vs$ on $[T_2, T]$. 
Formula~(\ref{e:ci-hereditary}) specialises as follows: 
\begin{equation}
\label{e:ci-nnpp}
\lambda_{n+1}( t | t_1, \dots, t_n ; \vr, \vs) =
\frac{ 
\pi(t-t_{n}) \pi( s_1 - t)
}{ 
\pi(s_1-t_{n})
} 1\{ t_{n} < t < T_2 \}
\end{equation}
with the convention that for $n=0$, $t_0 = \max(\vr)$ if $\vr \neq
\emptyset$ and $t_0 = 0$ otherwise. Compared to (\ref{e:ci-renewal}),
the survival probability $1 - F_\pi( T - t)$ is replaced by
$\pi( s_1 - t)$.  If $\vs = \emptyset$, (\ref{e:ci-nnpp}) reads
\[
\lambda_{n+1}( t | t_1, \dots, t_n ; \vr, \emptyset) =
\frac{ 
\pi(t-t_{n}) ( 1 - F_\pi(T-t) )
}{ 
1 - F_\pi(T-t_n)
} 1\{ t_{n} < t < T_2 \}
\]
as in (\ref{e:ci-renewal}), except for the fact that $t$ is restricted
to $(T_1, T_2)$. We conclude that the conditional distribution is a
nearest-neighbour Markov point process on $(T_1, T_2)$. Assuming
$\vs \neq \emptyset$, for $n\in \N$, 
its density can be written as
\begin{eqnarray*}
f( t_1, \dots, t_n | \vr, \vs ) & = & f(\emptyset | \vr, \vs) 
\frac{ 
\pi(t_1 - \max(\vr) ) \pi( s_1 - t_n)
}{
\pi( s_1 - \max(\vr) )
}
\prod_{i=2}^n  \pi(t_i - t_{i-1}) \\
& = & 
 f(\emptyset | \vr, \vs ) 
 \prod_{i=1}^n \varphi(t_i, \emptyset | \vr, \vs )
\prod_{i=2}^n \varphi( t_i, \{ t_{i-1} \} | t_1, \dots, t_n; \vr, \vs )
\end{eqnarray*}
with
\[
\varphi( t , \emptyset | \vr, \vs) = 
\frac{ \pi(t-\max(\vr)) \pi( s_1 -t ) }{ \pi( s_1 - \max(\vr) ) },
\]
under the convention $\max( \emptyset ) = 0$, and
\[
\varphi(u, \{ t \} | t_1, \dots, t_n; \vr, \vs ) =
\frac {  \pi( u - t ) \pi( s_1 - \max(\vr) ) }{
 \pi( u - \max(\vr) ) \pi( s_1 - t )}
\]
for $u \in \{ t_1, \dots, t_n \}$ and
 $t = \max\{ t_i : t_i < u  \} $ with
$\varphi^Y(t, \yy | t_1, \dots, t_n ) = 1$ otherwise.
For the special case $\vs = \emptyset$, the ratio 
\(
\pi( s_1 - t ) / \pi( s_1 - \max(\vr) )
\)
should be replaced by 
\(
 ( 1 - F_\pi(T-t) ) / ( 1 - F_\pi( T - \max(\vr) ).
\)
In conclusion, the boundary conditions $\vr$ and $\vs$ affect the
interaction functions only through the two nearest neighbours
$\max(\vr)$ and $\min(\vs)$, provided they exist.

\section{Cox processes}
\label{S:Cox}

A Cox process $Y$ is a two-stage stochastic process. More specifically,
let $\Lambda(t)$ be a random field on $[0,T]$ and define the
random measure $\Psi$ by
\[
\Psi(A) = \int_A \Lambda(t) dt
\]
for all Borel sets $A \subseteq [0,T]$.
Then, conditionally on $\Lambda = \lambda$, $Y$ is a chronologically 
ordered Poisson process on $[0,T]$ with intensity function
$\lambda$. Clearly, conditions must be imposed on $\Lambda$ so that 
the integral representation of $\Psi$ is well-defined, for example
that $\Lambda$ almost surely has continuous realisations 
\cite{Adle81,Adle07}.
The density and Janossy measures of $Y$ are found by integrating
out over the distribution of $\Lambda$, that is,
\begin{equation}
\label{e:jCox}
j_n^Y(t_1, \dots, t_n ) = \E\left[ 
e^{-\Psi([0,T])} 
\prod_{i=1}^n \Lambda(t_i) \right]
\end{equation}
for all $t_1 < \cdots < t_n$ \cite{Moll04}.

Since the Papangelou conditional intensity typically is not known
explicitly, we focus on Proposition~\ref{T:interpolation}.
Firstly, upon plugging in formula (\ref{e:jCox}), 
\[
f(t_1, \dots, t_n | \vr, \vs ) = 
\frac{e^{T_2-T_1}}{j^Z(\vr, \vs)} 
 \E_{\Lambda} \left[
   e^{-\Psi([0, T]}
   \prod_{ u \in (\vr, \vs)} \Lambda(u) 
   \prod_{ i=1}^n \Lambda(t_i) 
\right]
\]
with the marginal Janossy density of $Z$, the restriction
of $Y$ to $([0,T_1] \cup [T_2, T]) $, given by
\[
j^Z(\vr, \vs) = 
\E_{\Lambda}\left[
   \exp\left[ -\Psi( ([0, T_1] \cup [T_2, T])) \right]
   \prod_{ u \in (\vr, \vs)} \Lambda(u) 
   \right]
\]
since conditional on $\Lambda$, $Z$ is a Poisson process.
Finally, splitting
\(
\exp\left[ -\Psi([0, T]) \right] =
\exp\left[ -\Psi((T_1, T_2)) 
 -\Psi( ([0, T_1] \cup [T_2, T])) \right]
\),
we conclude that, provided $\vr$ and $\vs$ are feasible in the 
sense that $j^Z(\vr, \vs) > 0$, the conditional distribution of 
$Y$ on $(T_1, T_2)$, $0<T_1 < T_2 < T$, given $\vr$ and $\vs$ 
is a Cox process with its driving random field distributed as the
\[
   \prod_{ u \in (\vr, \vs) } \Lambda(u) 
\exp\left[ - \int_{[0,T]\setminus[T_1, T_2]} \Lambda(t) dt 
\right]
\]
weighting of the distribution of $\Lambda$.

As a clarifying example, consider the so-called compound Poisson 
process on $[0,T]$ for which $\Lambda$ is constant over $[0,T]$, 
taking either of the values
$\lambda_1 \neq \lambda_2$ with equal probability. Then, writing
$N$ for the sum of the cardinalities of $\vr$ and $\vs$ and
$\tilde T$ for the length of $[0,T_1] \cup [T_2, T]$,
\[
j_n(t_1, \dots, t_n | \vr, \vs ) = 
\frac{
\sum_i \lambda_i^{n+N} \exp\left[ -\lambda_i ( \tilde T ) 
 - \lambda_i ( T_2 - T_1 ) \right]
}{
\sum_i \lambda_i^{N} \exp\left[ -\lambda_i ( \tilde T ) \right]
}.
\]
In other words, the conditioning has the effect of changing the 
probability of $\lambda_i$ from a half to
\[
\frac{\lambda_i^N e^{-\lambda_i \tilde T} }{
\lambda_1^N e^{-\lambda_1 \tilde T}  + \lambda_2^N e^{-\lambda_2 \tilde T}  }.
\]
through $N$, the number of observed points. 
A further example will be studied in detail in Section~\ref{S:Diggle}.

\section{Simulation study:  renewal processes}
\label{S:renewal}

In this section, we will consider parameter and state estimation
for a well-known renewal process model. It is shown that naive approaches 
may result in bias. Instead, we shall adapt techniques for edge correction 
in spatial data to sequential point processes on the line. 

Figure~\ref{F:sample} shows a realisation of a temporal 
renewal process on $[0,4]$ observed within $[0,1] \cup [3,4]$ 
with Erlang inter-arrival probability density
\[
\pi(x) = \frac{\lambda^\alpha}{ (\alpha-1)!}  x^{\alpha-1} e^{-\lambda x}, 
\quad x \geq 0,
\]
for $\lambda = 40$ and $\alpha=2$.

\begin{figure}[htb]
\begin{center}
\centerline{
\epsfxsize=0.4\hsize
\epsfysize=0.4\hsize
\epsffile{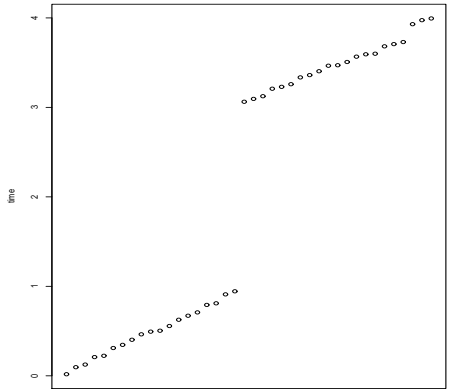}}
\end{center}
\caption{Sample from a renewal process with Erlang(2) inter-arrival times
and rate parameter $\lambda = 40$ observed in $[0,1] \cup [3,4]$. 
Time is plotted against index number.}
\label{F:sample}
\end{figure}
%postscript("Renewal.ps", horizontal=FALSE)
%plot(datapattern, xlab="", ylab="time", xaxt="n")
%\begin{verbatim}
% [1] 0.01763505 0.09626311 0.12557475 0.20933043 0.22381243 0.31064936
% [7] 0.34583165 0.40289698 0.46360943 0.49361227 0.50397257 0.55662029
%[13] 0.62666501 0.67152180 0.70879375 0.79226340 0.81026753 0.91014845
%[19] 0.94452381 3.06285848 3.09480290 3.12480238 3.20851280 3.22950155
%[25] 3.25900933 3.33509858 3.36089619 3.40400605 3.46560008 3.47104549
%[31] 3.50741535 3.56686498 3.59321276 3.59971261 3.68218634 3.70717226
%[37] 3.73018579 3.92970256 3.97464565 3.99375938
%\end{verbatim}

For renewal processes observed in an unbroken interval $[0,T]$, 
$T > 0$, \cite[Chapter~8]{Karr91} surveys two approaches to
estimate the parameters of the inter-arrival distribution $\pi$. 
The first method is to treat the fully observed inter-arrival times 
$L_i$, $i=1, \dots, N(T)$, as a random sample from $\pi$ and apply 
maximum likelihood or the method of moments to obtain parameter estimates. 
The second approach is to use the explicit representation of the 
likelihood in terms of the hazard rate $h^*$ and apply maximum 
likelihood estimation directly. Note that the first approach applies 
equally to observations on broken intervals, but the second one does not.

\subsection{Inference based on fully observed intervals}
\label{S:simulation}

Denote the lengths of the fully observed inter-arrival intervals by
$l_1, \dots, l_n$. If these would constitute a valid random sample,
for $\alpha$ fixed, the model would be an exponential family with 
sufficient statistic $\sum_i l_i$. The maximum likelihood estimator 
would exist and be given by
\(
\widehat{\lambda/\alpha} = n /  \sum_{i} l_i.
\)
The parameter $\alpha$ could be estimated by profile likelihood.
In our context, however, the sample size $n = N(T)$ is random.
Nevertheless, \cite{Karr91} suggests to proceed as if the 
$l_i$ were a random sample and claims this causes relatively 
little loss of information.
%page 104

For the realisation consisting of $40$ points depicted in 
Figure~\ref{F:sample}, consider the $38$ observable inter-arrival times.
For this sample, $\hat \lambda = 40.91$, whilst the profile likelihood 
method yields $\hat \alpha = 2$. 
To assess the bias and variance, we generated $100$ data patterns on 
$[0, 4]$ for the parameter values $\lambda=40$ and $\alpha=2$.
We iteratively excluded the middle of the left-most interval and 
estimated $\lambda$. The results are summarised in the table below.

\bigskip

\begin{center}
\begin{tabular}{|l||l|l|}
\hline
observation interval & mean  & variance \\
\hline
$[0,1] \cup [3,4]$ &  $42.1$ & $25.2$ \\
\hline
$[0, 0.25] \cup [ 0.75, 1]$ & $47.4$ & $136.2$\\
\hline
$[0, 0.0625] \cup [ 0.1875, 0.25]$ & $128.1$ & $1 \times 10^{4}$ \\
\hline
\end{tabular}
\end{center}

\bigskip

We conclude that the bias and variance increase as the observation
window contains less inter-arrival intervals.
The bias occurs as smaller intervals are more likely
to fall in the observation window, a phenomenon known as length
bias \cite{Karr91}. For the smallest interval, the bias is severe.
Indeed, of the samples on the union of two intervals of length $1/16$ each, 
less than half had observable inter-arrival times at all. Moreover,
the sub-interval lengths are only slightly larger than the expected
inter-arrival length of $1/20$.

\subsection{Monte Carlo maximum likelihood with missing data}
\label{S:Geyer}

The approach of Section~\ref{S:simulation} does not make full use of the
available data, as it completely ignores partially observed inter-arrival
intervals. Here, we adapt the method of \cite{Geye99} for dealing with 
missing data to better account for such intervals.

Again, suppose $Y$ is a renewal process on $[0,T]$ that is observed on 
$[0, T_1] \cup [T_2, T]$. Write $\pi = \pi_\theta$ for the density of the
inter-arrival time distribution. By part (b) of Theorem~\ref{T:interpolation}, 
$j^Z(\vr, \vs)$ is a normalising constant for the conditional Janossy
density on $(T_1, T_2)$ given the configuration elsewhere, 
say $\vr$ on $[0,T_1]$ and $\vs$ on $[T_2, T]$. Hence, writing $U$
for the vector of un-observed points in $Y$ falling in $(T_1, T_2)$, 
the log likelihood 
ratio with respect to a given reference parameter $\theta_0$ is
\cite{Geye99,Moll04}
\[
L(\theta)  =  \log j^Z_\theta(\vr, \vs) - \log j^Z_{\theta_0}(\vr, \vs)
 =  \log \E_{\theta_0} \left[
\frac{f_\theta^Y(U, Z)}{f_{\theta_0}^Y(U,Z)} \mid Z = (\vr, \vs) 
\right],
\]
where $f_\theta^Y$ is of the form (\ref{e:fY-renewal}).
In general, there is no explicit expression, so we must rely on Monte Carlo
approximation, that is, replace the expectation by an average over a 
sample from the conditional distribution to obtain
\begin{equation}
\label{e:Geyer}
L_N(\theta) = \log\left[ \frac{1}{N} \sum_{i=1}^N \frac{ 
f_\theta^Y(\vr, U_i^*, \vs)
}{
f_{\theta_0}^Y(\vr, U_i^*, \vs)
}\right].
\end{equation}
Specifically, the $U^*_i$ form a sample from the conditional distribution of 
$U$ on $(T_1, T_2)$ given $Z = (\vr, \vs)$ on $[0,T_1] \cup [T_2, T]$
under the reference parameter $\theta_0$. 
By taking derivatives, we get the Monte Carlo score and Fisher information. 
Note that even for exponential families, in the missing data case there 
may not be a unique maximum likelihood estimator.

To generate a sample $U^*_i$, $i=1, \dots, N$, we use the Metropolis--Hastings
approach \cite{Hand11,Geye94,Gree95,Lies06a} and tune it to the present 
context. An alternative is to use birth-and-death processes
\cite{Pres77}, but note that special care is needed to avoid explosion
while ensuring sufficient mixing at the same time. Further details are 
provided in the Appendix.

\begin{figure}[hbt]
\begin{center}
\centerline{
\epsfxsize=0.4\hsize
\epsfysize=0.4\hsize
\epsffile{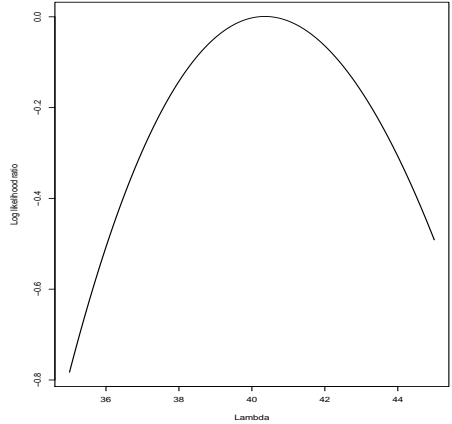}}
\end{center}
\caption{Monte Carlo log likelihood ratio for the data shown in
Figure~1 based on $10,000$ samples from a long run of the 
Metropolis--Hastings algorithm sub-sampled every thousand steps
with $\lambda_0 = 40.531$.}
\label{F:llr}
\end{figure}

Let us return to the model of Subsection~\ref{S:simulation} and
the data plotted in Figure~\ref{F:sample}. The parameter is the 
rate $\lambda$, and,  for $0 < t_1 < \cdots < t_n < T=4$,
\[
\frac{f^Y_{\lambda}(t_1, \dots, t_n)}{f^Y_{\lambda_0}(t_1, \dots, t_n)} = 
\left(\frac{\lambda}{\lambda_0}\right)^{ n \alpha }
\frac{
e^{-\lambda T} \sum_{i=0}^{\alpha -1} \lambda^i (T-t_n)^i / {i!}
}{
e^{-\lambda_0 T} \sum_{i=0}^{\alpha -1} \lambda_0^i (T-t_n)^i /{i!}
}.
\]
Hence, the sufficient statistic consists of $n$, the number of points, 
and $T-t_n$, the backward recurrence time from $T$. 
Write $\vr = (r_1, \dots, r_k)$ for the points observed in $[0, 1]$ 
and $\vs = (s_1, \dots, s_l)$ for those in $[3, 4]$. For our data,
$\vr$ consists of $19$ points, $\vs$ of $21$.  

To find a suitable
reference parameter for the log likelihood ratio, we use the 
Newton--Raphson method \cite{Geye99}.
For the data of Figure~\ref{F:sample}, this method gives 
$\lambda_0 = 40.531$. We compute the Monte Carlo log likelihood 
ratio $L_N(\lambda)$ using $N= 10,000$ samples from a long 
Metropolis--Hastings run sub-sampled every thousand steps after a
thousand steps burn-in. Note that it is sufficient to store the 
sufficient statistics only, in this case the length of the
vectors $U^*_i$.
The function is shown in Figure~\ref{F:llr}. It has a unique maximum  
at  $\hat \lambda = 40.36$; the Monte Carlo inverse Fisher information 
is $20.3$.

To assess the bias and variance, we considered a hundred data patterns
as in Section~\ref{S:simulation}. For each pattern, after ten Newton--Raphson 
steps from the true value $\lambda_0 = 40.0$, $N=1,000$ Monte 
Carlo samples were obtained by sub-sampling in a long
run every thousand steps after a burn-in of a thousand steps. 
The results are summarised in the table below.

\bigskip

\begin{center}
\begin{tabular}{|l||l|l|}
\hline
observation interval & mean ($\lambda$) & variance \\
\hline
$[0,1] \cup [3,4]$ & $40.9$ & $22.8$ \\
\hline
$[0, 0.25] \cup [ 0.75, 1]$ & $40.8$ & $76.7$ \\
\hline
$[0, 0.0625] \cup [ 0.1875, 0.25]$ & $40.8$ & $457.3$ \\
\hline
\end{tabular}
\end{center}

\bigskip

One sees that compared to the naive approach, the bias is
much reduced by correctly taking into account missing data and
does not increase noticeably for smaller sub-intervals.
The variance is also reduced but does increase when more data is
missing.

\section{Application}
\label{S:Diggle}

Figure~\ref{F:NHS} shows daily records of calls to NHS Direct,
a phone service operated by the National Health Service in Britain
which people could call 24 hours a day to get medical advice. The
calls shown are those that reported acute gastroenteric complaints 
in the county of Hampshire (on the English south coast and including
e.g.\ the towns of Winchester and Southampton). The full data were 
described and analysed in \cite{Digg05,Digg13,Digg03}. 
A modified version provided by Professor Diggle and Dr Hawtin
is available at
\begin{verbatim}
www.lancaster.ac.uk/staff/diggle/pointpatternbook/datasets/AEGISS.
\end{verbatim}
The latter data contain $7167$ records from the years 2001 and 2002.
Figure~\ref{F:NHS} plots the $327$ calls recorded during September 
and October 2001. A salient feature is that no calls were registered 
during the period September 13th--September 30th with recording 
resuming on October 1st, 2001. This gap is clearly visible in the
figure.

\begin{figure}
\begin{center}
\epsfxsize=0.4\hsize
\epsfysize=0.5\hsize
\epsffile{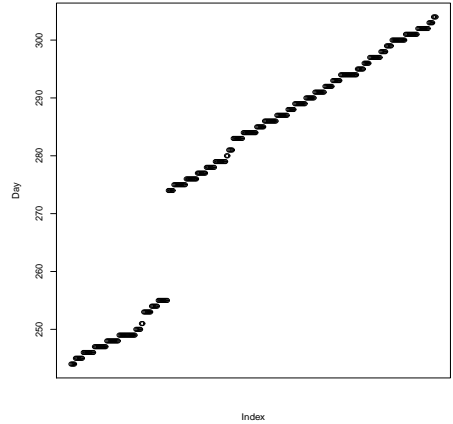}
\end{center}
\caption{Calls to NHS Direct from Hampshire in September--October
2001. The date of calls, in days starting January 1st, 2001, is 
plotted against index number.}
\label{F:NHS}
\end{figure}

As in \cite{Digg05}, we assume that calls are made according to a 
log-Gaussian Cox process on $[0, T]$ with $T=61$ (time measured in days) 
so that in the notation of Section~\ref{S:Cox}
\(
\Lambda(t) =  \mu_0(t) e^{S(t)}.
\)
The integrable function $\mu_0: [0,T] \to \R^+$ accounts for temporal 
variations  and $S$ is an Ornstein--Uhlenbeck process, 
that is, a stationary, Gaussian, Markovian process with 
exponentially decaying covariance function
\[
\mbox{Cov}(S(t), S(t^\prime)) = \sigma^2 e^{-\beta | t - t^\prime |}.
\]
The mean is chosen so that the expectation of $\exp( S(t) ) \equiv 1$.
By \cite{Adle81,Adle07,Moll98}, $S$ has almost surely continuous realisations
and the model is well-defined. 

Since the data is discretised in daily counts, the intensity of 
counts during the $i$-th day may be approximated by
\(
 \mu_0(i) e^{S(i)}
\),
$i = 0, \dots, 60$. As $S$ is an Ornstein--Uhlenbeck process, it
satisfies the Markov recursion
\begin{eqnarray*}
S(0)   & = & - \frac{\sigma^2}{2} + \sigma \Gamma(0)  \\
S(i) & = & - \frac{\sigma^2}{2} \left(  1 - e^{-\beta} \right) 
             + e^{-\beta} S(i-1) + \sigma \Gamma(i)
\end{eqnarray*}
for independent zero-mean normally distributed random variables $\Gamma(i)$
having variance $1 - \exp( - 2 \beta )$ under the convention
$1$ for $\Gamma(0)$. In terms of the $\Gamma(i)$,
\begin{equation}
\label{e:S}
S(i) = -\frac{\sigma^2}{2} + 
  \sigma \sum_{j=0}^{i} e^{-\beta (i-j)} \Gamma(j).
\end{equation}
In words, $S(i)$ is a sum of independent normally distributed
components discounted by elapsed time.

The goal of this section is two-fold. Firstly, we shall contrast 
moment based estimation with and without taking into account the
gap; then, we will consider state estimation.

\subsection{Parameter estimation}
\label{S:estNHS}

Although in principle it is possible to carry out likelihood based
inference, this would be computationally costly since the missing data 
consist not only of the missing counts in September 2001, but also 
of {\em all} the $S(i)$ or, equivalently, the $\Gamma(i)$, cf.\ 
Section~\ref{S:Cox}. Hence \cite{Digg05} proposed to use moment methods. 

\begin{figure}[hbt]
\begin{center}
\centerline{
\epsfxsize=0.4\hsize
\epsfysize=0.4\hsize
\epsffile{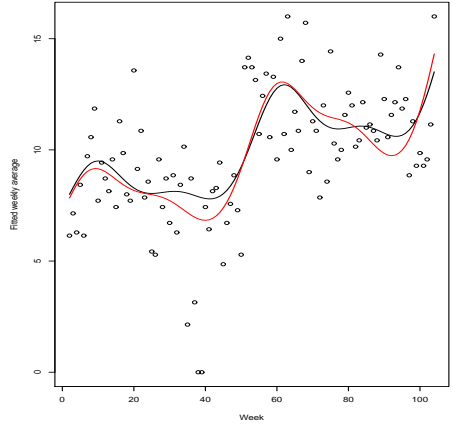}}
\end{center}
\caption{Weekly averages of calls to NHS Direct during 2001--20002 plotted
against the week index. Red line: fit by naive approach; Black: fit by 
approach taking into account the gap.}
\label{F:week-averages}
\end{figure}

First consider $\mu_0$. Since the expectation of $e^{S(i)}$ is one, 
$\mu_0(i)$ is the expected number of counts for day $i$. This observation
led \cite{Digg05} to fit a Poisson log-linear regression model of the form
\[
\log \mu_0(i) = \delta_{d(i)} + a_1 \cos \left( \frac{2 i \pi }{ 365} \right) +
b_1 \sin \left( \frac{2 i \pi }{ 365 }\right) +
a_2 \cos\left( \frac{4 i \pi }{ 365}\right) + 
b_2 \sin\left( \frac{4 i \pi }{ 365}\right) +  g i,
\]
$i=0, \dots, 60$. This model takes into account the apparent increase 
of incidences in Spring, the overall increase in calls to NHS Direct 
over time and the weekend-effect through the $d(i)$, but
%where the first day of the week is the Sunday.
ignores over-dispersion due to the Gaussian field. Note also that
the counts are conditionally rather than jointly independent.

To assess the effect of the gap, we estimate the parameters based
on two scenario's as follows. A naive estimator treats the zero counts in 
September as if they were true observations. The result is the 
red line in Figure~\ref{F:week-averages}. A more sophisticated approach
disregards the time period from September 13--30, with the 
black line in Figure~\ref{F:week-averages} as a result.
The trough due to the missing data is clearly visible in the
picture and, due to the periodic nature of the fitted model, 
also affects the behaviour a year later. 

\begin{figure}[hbt]
\begin{center}
\centerline{
\epsfxsize=0.4\hsize
\epsfysize=0.5\hsize
\epsffile{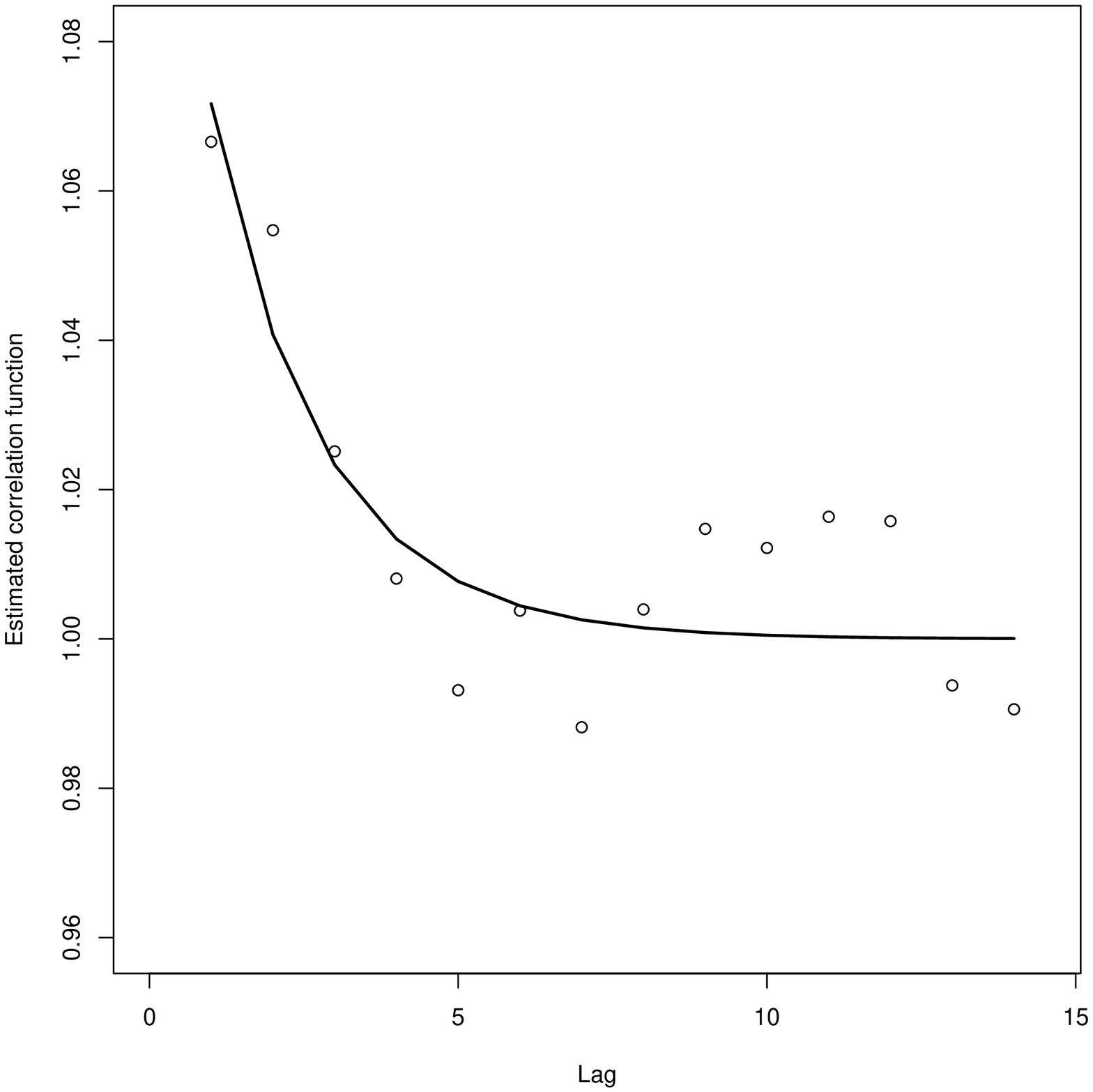}
\hspace{0.5cm}
\epsfxsize=0.4\hsize
\epsfysize=0.5\hsize
\epsffile{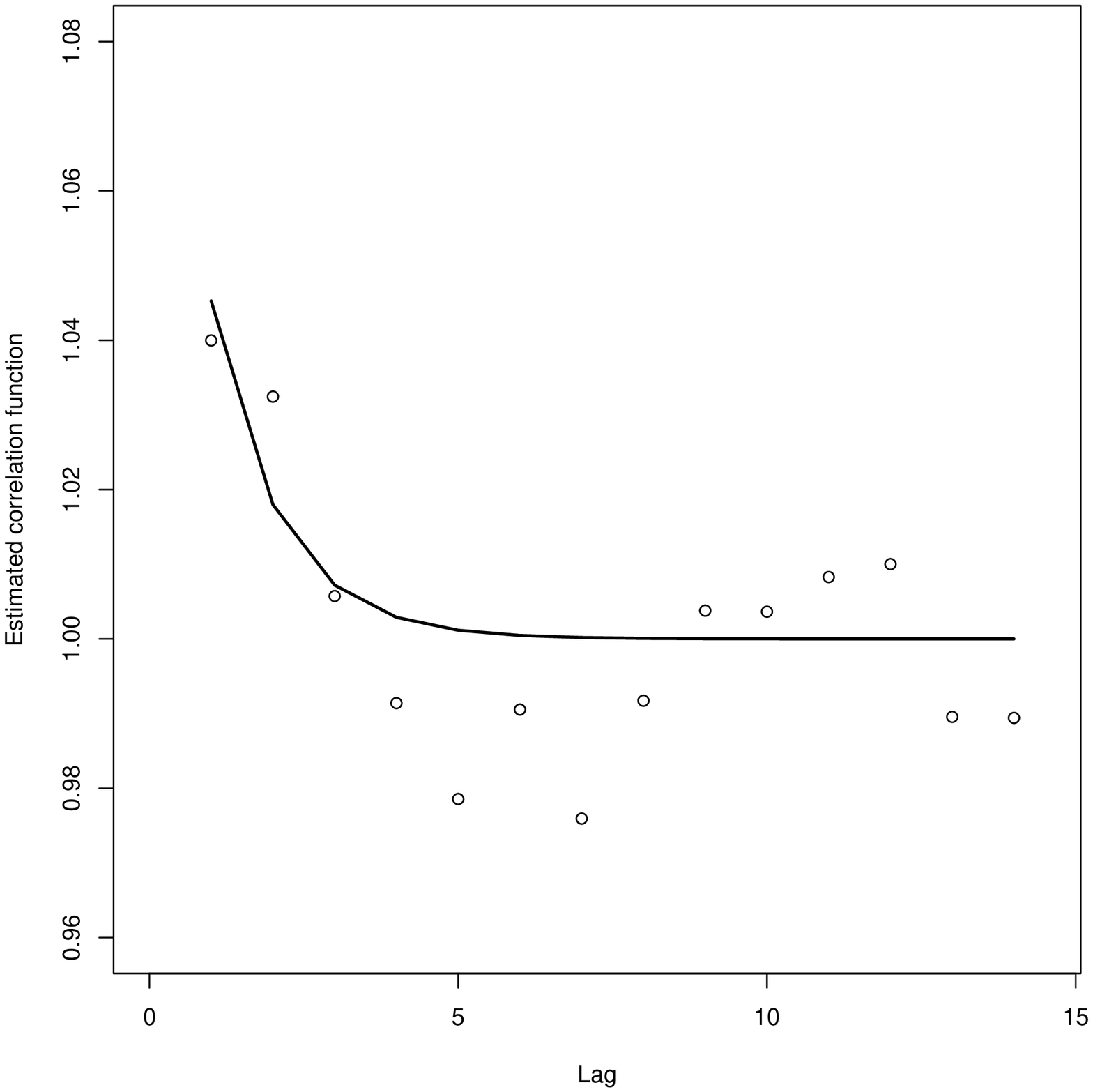}}
\end{center}
\caption{Minimum contrast fit for the naive approach (left) and when
taking into account the missing data (right). The points are the mean
of $N_i N_{i-\nu} / (\hat \mu_0(i) \hat \mu_0(i-\nu) )$ plotted as a 
function of the lag $\nu$. }
\label{F:mincontrast}
\end{figure}

Next, turn to estimation of $\sigma^2$ and $\beta$. Let $N_i$,
$i=0, \dots, I=729$, 
denote the observed number of calls during day $i$. The analogue
of the pair correlation function \cite{Chiu13} for count data is
\(
( N_i N_{i-\nu} ) / ( \mu_0(i) \mu_0(i-\nu) )
\)
which has expectation
\(
 \exp\left( \sigma^2 e^{-\beta | \nu| } \right) 
\)
at lag $\nu \neq 0$ regardless of $i$.
Therefore, the minimum contrast method minimises
\begin{equation}
\label{e:mincontrast}
\sum_{\nu=1}^m \left[ \frac{1}{I-\nu+1} \sum_{i=\nu}^I
\frac{ 
N_i N_{i-\nu} }
{
\hat \mu_0(i) \hat \mu_0(i-\nu) 
}
-
\exp\left( \sigma^2 e^{-\beta  \nu } \right) 
\right]^2
\end{equation}
over $\beta$ and $\sigma^2$. Here, $\hat \mu_0$ is an estimate of
$\mu_0$, $I = 729$ and $m$ is the number of lags considered. 

Minimising (\ref{e:mincontrast}) with $\hat \mu_0$ based on all 
data including the spurious zeroes using $m=14$ lags, we obtain
\(
\widehat{ \sigma^2} = 0.12
\)
and 
$\hat \beta = 0.55$.
The fit is indicated in the left-most panel of Figure~\ref{F:mincontrast}.
Note that $\mu_0$ is underestimated, hence the pair correlation 
function is overestimated.

The zero counts between September 13th and September 30th, 2001, should
not be taken into account. Thus, in (\ref{e:mincontrast}), we plug
in the appropriate estimator $\hat \mu_0$ (on which the black line in
Figure~\ref{F:week-averages} is based) and for lag $\nu$, consider only pairs 
$(N_i,  N_{i-\nu})$ for which both $i$ and $i-\nu$ do not fall in the 
period September 13--30. Doing so, we obtain $\widehat{\sigma^2} = 0.11$ 
and $\hat \beta = 0.91$. Thus, the value of $\sigma^2$
is not much affected, but that of $\hat \beta $ is. The fit is 
indicated in the right-most panel of Figure~\ref{F:mincontrast}.
We conclude that correlations between counts separated by a week or
more are small.

\subsection{State estimation}

Our final goal is to fill the gap in Figure~\ref{F:NHS} by sampling
from the conditional distribution given the data. Now, as
the estimated pair correlation function shown in Figure~\ref{F:mincontrast} 
is close to one for lags of a week and more, one may restrict
attention to the counts $n(T_i - 6), \dots, n(T_1)$ and 
$n(T_2), \dots, n(T_2 + 6)$ either side of the gap. Moreover, for
Cox processes, state estimation amounts to sampling from the driving
random measure, cf.\ Section~\ref{S:Cox}, which, by (\ref{e:S}), 
is uniquely defined by the $\Gamma(i)$, $i=0, \dots, T$. 
Therefore, the log likelihood, up to a constant
$c( \beta, \sigma^2, \mu_0, n(i)_i ) $, is given by
\begin{equation}
\label{e:condF}
- \frac{1}{2} \gamma(0)^2 
- \sum_{i=1}^T \frac{\gamma(i)^2}{2(1-e^{-2\beta })}
 + \sum_{i\in \{T_1 - 6, \dots, T_1, T_2, \dots, T_2 + 6 \} } \left[ n(i) S_\gamma(i) 
- \mu_0(i) e^{S_\gamma(i)} \right].
\end{equation}

\begin{figure}[htb]
\begin{center}
\centerline{
\epsfxsize=0.4\hsize
\epsfysize=0.5\hsize
\epsffile{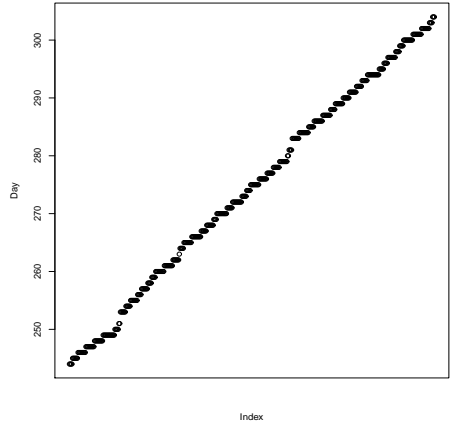}}
\end{center}
\caption{State estimation for the data shown in Figure~\ref{F:NHS}.
The date of calls is plotted against index number.}
\label{F:augmentation}
\end{figure}

\begin{figure}[hbt]
\begin{center}
\centerline{
\epsfxsize=0.4\hsize
\epsfysize=0.4\hsize
\epsffile{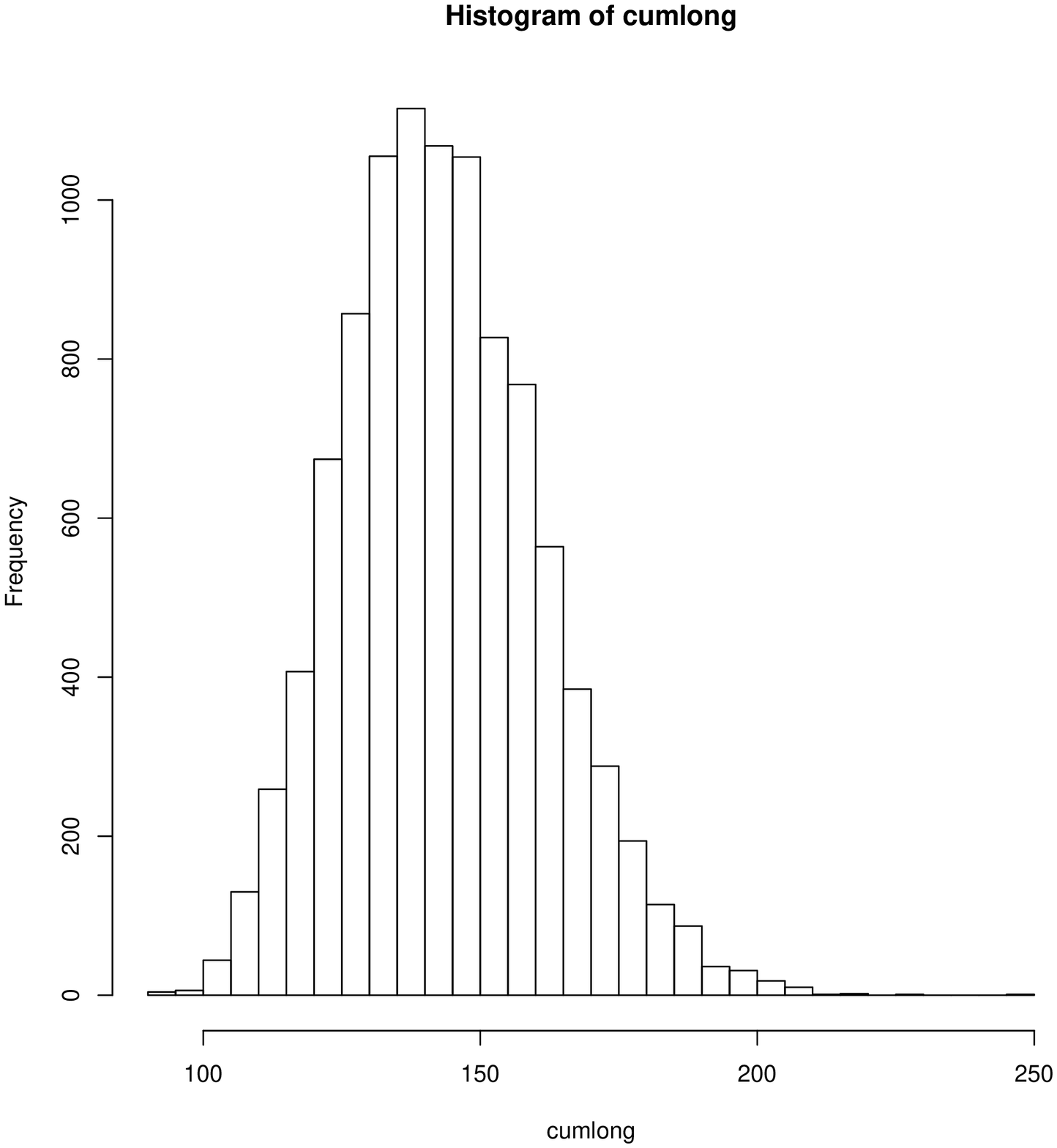}}
\end{center}
\caption{Histogram of the conditional total intensity of calls
during September 13--30, 2001, given the observed calls plotted
in Figure~\ref{F:NHS}.}
\label{F:hist}
\end{figure}

Since $c(\beta, \sigma^2, \mu_0, n(i)_i)$ cannot be evaluated exactly,
we use a Metropolis--Hastings algorithm to sample from (\ref{e:condF})
as in \cite{Brix01}. The proposal distribution
is a multivariate normal one with independent components having
variance $h>0$ and mean $h/2$ times the gradient of (\ref{e:condF}).
Thus, transitions are steered into a desirable direction.
It is shown in \cite{Robe97} that the Markov chain is
aperiodic and Lebesgue-irreducible, so that for almost all starting
states, the chain converges to (\ref{e:condF}) in total variation.

For $h = 0.5$, we ran the chain for $10,000$ steps. Using the 
realisation thus obtained, shown in Figure~\ref{F:augmentation},
we sampled the counts and re-estimated the parameters by minimum 
contrast. The results are similar to those obtained using the 
more sophisticated approach in Section~\ref{S:estNHS}. Indeed, 
$\widehat{\sigma^2} = 0.11$ and $\hat \beta = 0.93$, whereas the 
plot of $\hat \mu_0$ averaged over weeks is almost identical to the 
black line in Figure~\ref{F:week-averages}.

Next, we extended the run, sub-sampling every $1,000$ steps, to 
obtain $10,000$ realisations of (\ref{e:condF}) and calculated 
the histogram of the total intensity in the gap period. It is 
shown in Figure~\ref{F:hist}. In conclusion, around $150$ calls
may have been missed.

\section{Summary}

In this paper, we considered temporal point processes observed in
broken observation windows. We derived the marginal and conditional
distributions in terms of the Janossy densities of the underlying
sequential point process and studied Markovian and Cox model
in more detail. We carried out a simulation study to assess the
the length bias in renewal processes and analysed a real-life 
data set. The approach can easily be extended to space-time by
incorporating spatial marks. 

\section*{Acknowledgements}

Data were provided by Professor Peter Diggle, Lancaster University,
and Dr Peter Hawtin, Health Protection Agency, Southampton Laboratory.
The data were derived from anonymised data-sets collected during
project AEGISS, a collaborative surveillance project based in the south
of England.  The project was financially supported by the Food Standards 
Agency. The author acknowledges Professors Boucherie and Van der Mei
for fruitful discussions.

\clearpage
\section*{Appendix: Markov chain Monte Carlo}

To simulate a point process whose distribution is defined by the
Papangelou conditional intensity, we adapt two general strategies 
to our context. Throughout, let $Y$ be an hereditary sequential
point process on $[0, T]$ with density $f$ with respect to 
$\nu_{[0,T]}$ and Papangelou conditional intensity 
$\lambda_i(\cdot | \cdot)$. 

\subsection*{Metropolis--Hastings sampling}

The Metropolis--Hastings method works by proposing an update according to a
distribution that is convenient to sample from, and then to accept or
reject this proposal with a probability that is chosen so as to make
sure the detailed balance equations are satisfied \cite{Hand11}.

The two generic types of proposals are births and deaths. More precisely, 
with probability $1/2$ propose a birth, otherwise a death. In the first 
case, select a point $u$ uniformly on $[0,T]$ and insert it in its 
chronological position $i$. In case of a death, a point is chosen 
uniformly for deletion; if the sequence is empty, nothing happens.
Then, if the current vector is $(t_1, \dots, t_n)$, the proposal to
add $u$ is accepted with probability
\[
\min \left\{ 1,
\frac{ \lambda_i( u | t_1, \dots, t_n ) T  }{n + 1}
\right\};
\]
the proposal to delete the $i$-th point with probability
\[
\min \left\{ 1, \frac{ n} {
 \lambda_i( t_i | t_1, \dots, t_{i-1}, t_{i+1}, \dots, t_n ) T  }
\right\}.
\]

It is easily seen that $f(\cdot)$ is an invariant density. Moreover, 
if we start the chain in a sequence $\vt$ for which $f(\vt) > 0$, 
the chain will almost surely never leave the set of states having positive
density. 

In order to show that the Metropolis--Hastings chain $Y_n$, $n\in \N_0$,
converges to $f(\cdot)$ in total variation from any initial state having
positive density, it is sufficient to assume stability, that is, that
the Papangelou conditional intensity is uniformly bounded by some 
$\beta > 0$. The proof follows the same lines as in \cite{Geye99}.

\subsection*{Birth-and-death process sampling}

Since the Papangelou conditional intensity may fluctuate a lot, 
we define birth-and-death processes in terms of local bounds. 
More precisely, assume that 
\begin{equation}
\label{e:local-bound}
\lambda_i( u | t_1, \dots, t_n ) \leq g(u | t_1, \dots, t_n )\leq \beta
\end{equation}
for some integrable function $g$ that is constant on 
$(t_{i-1}, t_i)$, $i=1, \dots, n+1$, say $g_i(t_1, \dots, t_n)$, and
some $\beta > 0$, with an appropriate convention for $i=1$ and $i=n+1$. 
Then, the birth-and-death algorithm runs as follows. If the current state 
is $\vt = (t_1, \dots, t_n)$, 
\begin{itemize}
\item compute the upper bound $G(t_1, \dots, t_n)$ to the total birth rate by
\[
 G(t_1, \dots, t_n ) = \sum_{i=1}^{n+1} \int_{t_{i-1}}^{t_i} 
   g(u | t_1, \dots, t_n) du = 
\sum_{i=1}^{n+1} ( t_i - t_{i-1} ) g_i( t_1, \dots, t_n)
\]
and take death rate $D(t_1, \dots, t_n) = n$;
\item generate an exponentially distributed sojourn time with rate 
parameter $G(\vt) + D(\vt)$;
\item with probability 
\(
\frac{G(\vt)}{ G(\vt) + D(\vt) },
\)
generate a new point (`birth') as follows:
\begin{itemize}
\item sample an interval with probability 
\(
\frac{ (t_i-t_{i-1}) g_i(\vt) }{ G(\vt) }
\)
and propose a new point $u$ uniformly in the chosen interval;
\item  accept the proposal with probability 
\(
\frac{ \lambda_i(u | \vt) }{ g_i(\vt)}
\);
\end{itemize}
\item with probability 
\(
\frac{D(\vt)}{ G(\vt) + D(\vt) },
\)
delete a uniformly chosen point.
\end{itemize}

Then the rate for a transition from $\vt$ to $(t_1, \dots, t_i, u,
t_{i+1}, \dots, t_n)$ for $t_i < u < t_{i+1}$ is
\[
( G(\vt) + D(\vt) ) \frac{ G(\vt) }{ G(\vt) + D(\vt) } 
\frac{(t_i - t_{i-1}) g_i(\vt) }{ G(\vt) } \frac{1}{t_i - t_{i-1}}
\frac{\lambda_i(u | \vt) }{ g_i(\vt)} = \lambda_i(u | \vt)
\]
so that the birth rate is effectively $\lambda_i( u | t_1, \dots, t_n)$ 
on $(t_{i-1}, t_i)$. Therefore, $f$ is an invariant density. An appeal to 
\cite{Pres77} implies the existence of a unique jump
process with the given birth and death rates; $f$ is its unique 
invariant probability density and the process converges to $f$ in 
distribution from any $\vt$ for which $f(\vt) > 0$. 

\paragraph{Renewal process}

As an example, consider a renewal process with Erlang inter-arrival
times. To derive a local bound on the conditional intensity, note that
for $a < \xi < b$, 
\[
\frac{
\pi(\xi - a)\pi(b-\xi)
}{
\pi(b-a)
} \leq 
    \frac{ \lambda^\alpha }{ ( \alpha - 1 ) !}
     \left( \frac{ b-a }{4} \right)^{\alpha - 1}
\]
and 
\[
\frac{\pi(\xi-a) ( 1 - F_\pi(b-\xi) )
}{1 - F_\pi(b-a)}
=
\frac{ \lambda^\alpha }{ ( \alpha - 1 ) !}
(\xi-a)^{\alpha-1} \frac{ \sum_{i=0}^{\alpha-1}
     \lambda^i (b - \xi)^i / i! }{
 \sum_{i=0}^{\alpha-1}
\lambda^i (b-a)^{i} / i!}
\]
is bounded from above by
\[
 \frac{\lambda^\alpha (b-a)^{\alpha - 1} }{
 (\alpha -1)! \sum_{i=0}^{\alpha -1} \lambda^i (b-a)^i/i!
}
 \sum_{i=0}^{\alpha -1} \left[ 
       \frac{\lambda (b-a)}{4} \right]^i/i!.
\]
The last bound follows from the fact that 
\begin{eqnarray*}
(\xi - a)^{\alpha - 1} (b-\xi)^i & =  &
(\xi - a)^{\alpha-1-i} \left[ (b-\xi) (\xi-a) \right]^i \\
& \leq &
(b - a)^{\alpha-1-i}  \left[ \frac{(b-a)^2}{4} \right]^i 
= 
(b - a)^{\alpha-1+i} 4^{-i} .
\end{eqnarray*}
The bound can be improved upon in some cases by noting that 
the function
\[
\phi: \xi \to \pi( \xi - a) (1 - F_\pi( b - \xi))
\] 
increases on $(a,b)$ if
%attains its maximum at the point $\xi$ having hazard at 
%$b-\xi$ equal to $\lambda - (\alpha-1)/(\xi-a)$, if 
$\lambda - (\alpha-1) / (b-a) < 0$.
Finally, since the sub-interval length $b-a \leq T$ is bounded,
stability follows.

\end{document}